\documentclass[preprint]{elsarticle}

\usepackage{lineno,hyperref}
\modulolinenumbers[5]

\journal{Physica D: Nonlinear Phenomena}

\usepackage[utf8]{inputenc}
\usepackage[english]{babel}
\usepackage{amsmath}
\usepackage{amsfonts}
\usepackage{amsthm}
\usepackage[numbers]{natbib}
\usepackage[export]{adjustbox}
\newtheorem{theorem}{Theorem}

\newtheorem{definition}{Definition}
\newtheorem{assumption}{Assumption}
\renewcommand{\theassumption}{(A\arabic{assumption})}

\newtheorem{remark}{Remark}
\newtheorem{lemma}{Lemma}

\DeclareMathOperator{\sech}{sech}
\DeclareMathOperator{\arctanh}{arctanh}

\DeclareMathOperator{\antidiag}{antidiag}

\begin{document}

\begin{frontmatter}

\title{Krein signature for instability of $\mathcal{PT}$-symmetric states}

\author[mcmaster]{Alexander Chernyavsky\fnref{note1}}
\author[mcmaster,nizhny]{Dmitry E. Pelinovsky\fnref{note2}\corref{cor1}}
\address[mcmaster]{\small Department of Mathematics, McMaster University, Hamilton, Ontario, L8S 4K1, Canada}
\address[nizhny]{\small Department of Applied Mathematics, Nizhny Novgorod State Technical University,
Nizhny Novgorod 603950, Russia }
\fntext[note1]{e-mail: chernya@math.mcmaster.ca}
\fntext[note2]{e-mail: dmpeli@math.mcmaster.ca}
\cortext[cor1]{Corresponding author}

\begin{abstract}
{\em Krein quantity} is introduced for isolated neutrally stable eigenvalues associated with the
stationary states in the $\mathcal{PT}$-symmetric nonlinear Schr\"{o}dinger equation. Krein quantity
is real and nonzero for simple eigenvalues but it vanishes if two simple eigenvalues coalesce
into a defective eigenvalue. A necessary condition for bifurcation of unstable eigenvalues
from the defective eigenvalue is proved. This condition requires the two simple eigenvalues
before the coalescence point to have {\em opposite} Krein signatures. The theory is illustrated with
several numerical examples motivated by recent publications in physics literature.
\end{abstract}


\end{frontmatter}

\section{Introduction}

Dynamical systems are called $\mathcal{PT}$-symmetric if they are invariant with
respect to the combined parity ($\mathcal{P}$) and time-reversal ($\mathcal{T}$) transformations.
A non-Hermitian $\mathcal{PT}$-symmetric linear operator may have a real spectrum
and may define a unitary time evolution of the linear $\mathcal{PT}$-symmetric  system \cite{bender}. A
non-Hamiltonian $\mathcal{PT}$-symmetric nonlinear system may have a continuous family
of stationary states parameterized by their energy \cite{ptreview, suchkov}.

Originated in quantum mechanics~\cite{bender2005,mostafazadeh}, the topic of
$\mathcal{PT}$-symmetry was later boosted by applications in optics~\cite{makris, musslimani}
and other areas of physics~\cite{benderExp,rubinsteinExp,Schindler}.
Recent applications include single-mode $\cal{PT}$ lasers~\cite{feng14,hodaei}
and unidirectional reflectionless $\cal{PT}$-symmetric metamaterials at
optical frequencies~\cite{feng13}.

The non-Hermitian $\mathcal{PT}$-symmetric linear operator may lose real eigenvalues at
the so-called $\cal{PT}$-phase transition point, where two real eigenvalues coalesce and
bifurcate off to the complex plane, creating instability. A stationary state
of the non-Hamiltonian $\mathcal{PT}$-symmetric nonlinear system may
exist beyond the $\mathcal{PT}$-phase transition point but may become spectrally unstable
due to coalescence of purely imaginary eigenvalues and their bifurcation off to the complex plane.
Examples of such instabilities have been identified for many $\mathcal{PT}$-symmetric
linear and nonlinear systems \cite{bender,ptreview, suchkov}.

In Hamiltonian systems, instabilities arising due to coalescence of purely imaginary eigenvalues
can be predicted by computing the {\em Krein signature} for each eigenvalue, which is defined as
the sign of the quadratic part of Hamiltonian restricted to the associated eigenspace of the linearized problem.
When two purely imaginary eigenvalues coalesce,
they bifurcate off to the complex plane only if they
have opposite Krein signatures prior to collision~\cite{kapitulabook}.
The concept of Krein signature was introduced by MacKay~\cite{mackay} in the case of
finite-dimensional Hamiltonian systems, although the idea dates back to the works of Weierstrass~\cite{weierstrass}.

There have been several attempts to extend the concept of Krein signature
to the non-Hamiltonian $\mathcal{PT}$-symmetric systems.
Nixon and Yang~\cite{yang} considered the linear Schr\"{o}dinger equation
with a complex-valued $\mathcal{PT}$-symmetric potential and introduced
the indefinite $\mathcal{PT}$-inner product with the induced $\mathcal{PT}$-Krein signature,
in the exact correspondence with the Hamiltonian-Krein signature.
In our previous works \cite{CP1,CP2}, we considered a Hamiltonian
version of the $\mathcal{PT}$-symmetric system of coupled oscillators
and introduced Krein signature of eigenvalues by using the corresponding Hamiltonian.
In the recent works \cite{AB,PTL,SS}, a coupled non-Hamiltonian $\mathcal{PT}$-symmetric system
was considered and the linearized system was shown to be block-diagonalizable
to the form where Krein signature of eigenvalues can be introduced.
All these cases were too special, the corresponding Krein signatures
cannot be extended to a general $\mathcal{PT}$-symmetric system.

In this work, we address the following nonlinear Schr\"{o}dinger's equation (NLSE) with
a general complex potential:
\begin{equation}
i\partial_t \psi + \partial^2_x \psi - ( V(x) + i\gamma W(x) )\psi + g|\psi|^2 \psi = 0,
\label{NLS}
\end{equation}
where $\gamma\in\mathbb{R}$ is a gain-loss parameter, $g=+1$ ($g=-1$) defines
focusing (defocusing) nonlinearity, and the real potentials $V$ and $W$
satisfy the even and odd symmetry, respectively:
\begin{equation}
\label{potentials}
V(x) = V(-x), \quad W(-x) = -W(x), \quad x \in \mathbb{R}.
\end{equation}
In quantum physics, the complex potential $V + i \gamma W$ is used to
describe effects observed when quantum particles are loaded in an open
system~\cite{wunner,dast}. The intervals with
positive and negative imaginary part correspond to the gain and loss of quantum particles,
respectively. When gain exactly matches loss, which happens under the symmetry
condition (\ref{potentials}), the potential $V + i \gamma W$ is
 $\cal{PT}$-symmetric with respect to the parity operator $\mathcal{P}$ and
the time reversal operator $\cal{T}$ acting on a function $\psi(x,t)$ as follows:
\begin{equation}
\label{operators}
    \mathcal{P} \psi(x,t) = \psi(-x,t), \quad
    \mathcal{T} \psi(x,t) = \overline{\psi(x,-t)}.
\end{equation}
The NLSE (\ref{NLS}) is $\mathcal{PT}$-symmetric under the condition (\ref{potentials})
in the sense that if $\psi(x,t)$ is a solution to (\ref{NLS}), then
\begin{equation*}
\widetilde{\psi}(x,t) = \mathcal{PT} \psi(x,t) = \overline{\psi(-x,-t)}
\end{equation*}
is also a solution to (\ref{NLS}).

The NLSE (\ref{NLS}) with a $\mathcal{PT}$-symmetric potential is also used in the paraxial nonlinear optics.
In that context, time and space have a meaning of longitudinal and transverse coordinates,
and complex potential models the complex refractive index~\cite{delgado}.
Another possible application of the NLSE (\ref{NLS})
is Bose-Einstein condensate, where it models the dynamics of the self-gravitating boson gas trapped in a
confining potential $V$. Intervals, where $W$ is positive and negative,
allow one to compensate atom injection and particle leakage, correspondingly~\cite{wunner}.

Here we deal with the stationary states in the NLSE (\ref{NLS}) and introduce
Krein signature of isolated eigenvalues in the spectrum of their linearization.
We prove that the necessary condition for the onset
of instability of the stationary states from a defective eigenvalue of algebraic multiplicity two
is the {\em opposite} Krein signature of the two simple isolated eigenvalues prior to
their coalescence. Compared to the Hamiltonian system in \cite{CP1}
or the linear Schr\"{o}dinger equation in \cite{yang}, the Krein signature
of eigenvalues cannot be computed from the eigenvectors in the linearized
problem, as the adjoint eigenvectors need to be computed
separately and the sign of the adjoint eigenvector needs to be chosen by a continuity argument.

We show how to compute Krein signature numerically for several examples of
the $\cal{PT}$-symmetric potentials. In the focusing case $g = 1$, we consider the Scarf II
potential studied in~\cite{ahmed,wadati,Kev,yang} with
\begin{equation}
\label{pot-Wadati}
V(x) = -V_0 \sech^2(x), \quad
W(x) = \sech(x)\tanh(x),
\end{equation}
where $V_0 > 0$ is a parameter. This potential is a complexification of the real hyperbolic 
Scarf potential~\cite{Bagchi}. The nonhyperbolic version of the latter first appeared in~\cite{Scarf}, 
where the linear Schr\"{o}dinger equation with Scarf potential was solved.
In the defocusing case $g = -1$,
we consider the confining potential studied in
\cite{kevrekidis} with
\begin{equation}
\label{pot-BEC}
V(x) = \Omega^2 x^2, \quad W(x) = xe^{-\frac{x^2}{2}},
\end{equation}
where $\Omega > 0$ is a parameter. In agreement with the theory,
we show that the coalescence of two isolated eigenvalues
in the linearized problem associated with the stationary states in the NLSE (\ref{NLS})
leads to instability only if the Krein signatures of the
two eigenvalues are opposite to each other.

The paper is organized as follows. Section \ref{sec-theory} introduces the stationary states,
eigenvalues of the linearization, and the Krein signature of eigenvalues for the NLSE (\ref{NLS})
under some mild assumptions. Section \ref{sec-proof} gives the proof of the necessary condition
for the instability bifurcation from a defective eigenvalue of algebraic multiplicity two.
Section \ref{sec-numerics} explains details of the numerical technique.
Section \ref{sec-examples} presents outcomes of numerical approximations
for the two potentials (\ref{pot-Wadati}) and (\ref{pot-BEC}). Section \ref{sec-conclusion}
concludes the paper with open questions.

\vspace{0.25cm}

{\bf Acknowledgements.} We thank P.G. Kevrekidis for suggesting the problem back in 2012 and for useful discussions.
A. Chernyavsky is supported by the McMaster graduate scholarship. D.E. Pelinovsky is supported from
the state task of Russian Federation in the sphere of scientific activity (Task No. 5.5176.2017/8.9).

\section{Stationary states, eigenvalues, and Krein signature}
\label{sec-theory}

Let us define the stationary state of the NLSE (\ref{NLS}) by $\psi(x,t) = \Phi(x) e^{-i \mu t}$, where $\mu \in \mathbb{R}$
is a parameter. In the context of BECs, $\mu$ has the meaning of the chemical potential~\cite{dast}.
The function $\Phi(x) : \mathbb{R} \to \mathbb{C}$ is
a suitable solution of the stationary NLSE in the form
\begin{equation}
-\Phi''(x) + ( V(x) + i\gamma W(x) ) \Phi(x) - g |\Phi(x)|^2 \Phi(x) = \mu \Phi(x), \quad x \in \mathbb{R}.
\label{NLSstat}
\end{equation}
We say that $\Phi$ is a $\mathcal{PT}$-symmetric stationary state if $\Phi$ satisfies the $\mathcal{PT}$ symmetry:
\begin{equation}
\label{PT-sym}
\Phi(x) = \mathcal{PT} \Phi(x) = \overline{\Phi(-x)}, \quad x \in \mathbb{R}.
\end{equation}
In addition to the symmetry constraints on the potentials $V$ and $W$ in (\ref{potentials}),
our basic assumptions are given below. Here and in what follows, we denote the Sobolev space of square integrable functions
with square integrable second derivatives by $H^2(\mathbb{R})$ and the weighted
$L^2$ space with a finite second moment by $L^{2,2}(\mathbb{R})$.

\begin{assumption}
\label{assumption-1}
We assume that the linear Schr\"{o}dinger operator $L_0 := -\partial_x^2 + V(x)$ in $L^2(\mathbb{R})$
admits a self-adjoint extension with a dense domain $D(L_0)$ in $L^2(\mathbb{R})$.
\end{assumption}

\begin{remark}
If $V \in L^2(\mathbb{R}) \cap L^{\infty}(\mathbb{R})$ as in (\ref{pot-Wadati}),
then Assumption~\ref{assumption-1} is satisfied with $D(L_0) = H^2(\mathbb{R})$
(see \cite{sigal}, Ch. 14, p.143).
If $V$ is harmonic as in (\ref{pot-BEC}),
then Assumption~\ref{assumption-1} is satisfied with
$D(L_0) = H^2(\mathbb{R}) \cap L^{2,2}(\mathbb{R})$ (see ~\cite{heffler}, Ch. 4, p.37).
\end{remark}

\begin{assumption}
\label{assumption-2}
We assume that $W$ is a bounded and exponentially decaying potential satisfying
\begin{equation*}
|W(x)| \leq C e^{-\kappa |x|}, \quad x \in \mathbb{R},
\end{equation*}
for some $C > 0$ and $\kappa > 0$.
\end{assumption}

\begin{remark}
Both examples in (\ref{pot-Wadati}) and (\ref{pot-BEC}) satisfy Assumption~\ref{assumption-2}.
By Assumption~\ref{assumption-2}, the potential $i \gamma W$ is
a relatively compact perturbation to $L_0$ (see \cite{reed4}, Ch. XIII, p.113).
This implies that the continuous spectrum of $L_0 + i\gamma W$ is the same as $L_0$.
If $V \in L^2(\mathbb{R}) \cap L^{\infty}(\mathbb{R})$, then the continuous spectrum of $L_0$
is located on the positive real line. If $V$ is harmonic, then the continuous spectrum of $L_0$
is empty (see~\cite{reed4}, Ch. XIII, Theorem 16 on p.120).
\end{remark}

\begin{assumption}
\label{assumption-3}
We assume that for a given $\mu \in \mathbb{R}$, there exist $\gamma_* > 0$ and
a bounded, decaying, and $\mathcal{PT}$-symmetric solution $\Phi \in D(L_0) \subset L^2(\mathbb{R})$
to the stationary NLSE (\ref{NLSstat}) with $\gamma \in (-\gamma_*,\gamma_*)$ satisfying (\ref{PT-sym}) and
\begin{equation*}
|\Phi(x)| \leq C e^{-\kappa |x|}, \quad x \in \mathbb{R},
\end{equation*}
for some $C > 0$ and $\kappa > 0$. Moreover, the map $(-\gamma_*,\gamma_*) \ni \gamma \mapsto \Phi \in D(L_0)$
is real-analytic.
\end{assumption}

\begin{remark}
Since the nonlinear equation~\eqref{NLSstat} is real-analytic in $\gamma$, the Implicit Function Theorem
(see~\cite{Zeidler}, Ch. 4, Theorem 4.E on p.250) provides real analyticity of the map
 $(-\gamma_*,\gamma_*) \ni \gamma \mapsto \Phi \in D(L_0)$ as long as the Jacobian operator
\begin{equation}
\label{jacobian}
\mathcal{L} := \left[\begin{array}{cc}
      -\partial_x^2 + V + i \gamma W - \mu - 2g |\Phi|^2 & -g\Phi^2 \\
      -g\overline\Phi^2 & -\partial_x^2 + V - i \gamma W - \mu - 2g |\Phi|^2
\end{array}\right]
\end{equation}
is invertible in the space of $\mathcal{PT}$-symmetric functions in $L^2(\mathbb{R})$.
\end{remark}

\begin{remark}
Under Assumption~\ref{assumption-3}, we think about $\mu$ as a fixed parameter and $\gamma$ as a varying parameter
in the interval $(-\gamma_*,\gamma_*)$. The interval includes the Hamiltonian case $\gamma = 0$.
In the context of the example of $V$ in (\ref{pot-Wadati}), it will be more natural to fix the value of $\gamma$
and to consider the parameter continuation of $\Phi \in D(L_0)$ with respect to $\mu$.
The results are analogous to what we present here under Assumption \ref{assumption-3}.
\end{remark}

We perform the standard linearization of the NLSE (\ref{NLS})
near the stationary state $\Phi$ by substituting
\begin{equation*}
\psi(x,t) = e^{-i\mu t} \left[ \Phi(x) + u(t,x) \right]
\end{equation*}
into the NLSE ~\eqref{NLS} and truncating at the linear terms in $u$:
\begin{equation*}
\begin{cases}
i u_t = (-\partial^2_x + V + i\gamma W - \mu - 2g|\Phi|^2)u - g\Phi^2 \overline{u}, \\
-i \overline{u}_t = (-\partial^2_x + V - i\gamma W - \mu - 2g|\Phi|^2)\overline{u} - g\overline\Phi^2 u.
\end{cases}
\end{equation*}

Using $u = Ye^{-\lambda t}$ and $\overline{u} = Ze^{-\lambda t}$ with the spectral parameter $\lambda$
yields the spectral stability problem in the form
\small
\begin{equation}
        \arraycolsep=0pt\def\arraystretch{1}
\mathcal{L} \left[\begin{array}{c} Y \\ Z \end{array}\right] = -i\lambda\sigma_3
\left[\begin{array}{c} Y \\ Z \end{array}\right],
\label{original}
\end{equation}
\normalsize
where $\sigma_3 = {\rm diag}(1,-1)$ is the third Pauli's matrix and $\mathcal{L}$
is given by (\ref{jacobian}). Note that if $\lambda\not\in\mathbb{R}$, then $Z \ne \overline{Y}$.

\begin{lemma}
\label{lemma-continuous-spectrum}
The continuous spectrum of the operator $i \sigma_3 \mathcal{L} : D(L_0) \times D(L_0) \to L^2(\mathbb{R}) \times L^2(\mathbb{R})$,
if it exists, is a subset of $i \mathbb{R}$.
\end{lemma}

\begin{proof}
Thanks to the Assumptions~\ref{assumption-1},~\ref{assumption-2} and~\ref{assumption-3},
$W$ and $\Phi^2$ terms are relatively compact perturbations to the diagonal unbounded
operator $\mathcal{L}_0 := {\rm diag}(L_0-\mu I,L_0- \mu I)$, where
$L_0 = -\partial^2_x + V$ is introduced in (A1) and $I$ is an identity
$2\times 2$ matrix. Therefore,
\begin{equation*}
\sigma_c(i \sigma_3 \mathcal{L}) = \sigma_c(i \sigma_3 \mathcal{L}_0) \subset i \mathbb{R},
\end{equation*}
where $\sigma_c(A)$ denotes the absolutely continuous part of the spectrum of the operator
$A : D(A) \subset L^2(\mathbb{R}) \to L^2(\mathbb{R})$.
\end{proof}

\begin{remark}
If $V \in L^2(\mathbb{R}) \cap L^{\infty}(\mathbb{R})$, then $\mu < 0$ and
\begin{equation*}
\sigma_c(i \sigma_3 \mathcal{L}) = i (-\infty,-|\mu|] \cup i[|\mu|,\infty).
\end{equation*}
If $V$ is harmonic, then $\sigma_c(i \sigma_3 \mathcal{L})$ is empty.
\end{remark}

\begin{definition}
We say that the stationary state $\Phi$ is spectrally stable if every nonzero solution
$(Y,Z) \in D(L_0) \times D(L_0)$ to the spectral problem (\ref{original})
corresponds to $\lambda \in i \mathbb{R}$.
\end{definition}

We note the quadruple symmetry of eigenvalues in the spectral problem (\ref{original}).

\begin{lemma}
\label{lem-symmetry}
If $\lambda_0$ is an eigenvalue of the spectral problem ~\eqref{original},
so are $-\lambda_0$, $\bar\lambda_0$, and $-\bar\lambda_0$.
\end{lemma}

\begin{proof}
We note the symmetry of $\mathcal{L}$ and $\sigma_3$:
\begin{equation}
\label{symm1}
\mathcal{L} = \sigma_1 \overline{\mathcal{L}} \sigma_1, \quad
\sigma_3 = -\sigma_1 \sigma_3 \sigma_1,
\end{equation}
where $\sigma_1 = \antidiag(1,1)$ is the first Pauli's matrix.
If $\lambda_0$ is an eigenvalue of the spectral problem~\eqref{original} with the eigenvector
$v_0 := (Y,Z)$, then so is $\overline{\lambda}_0$ with the eigenvector
$\sigma_1 \overline{v}_0 = (\overline{Z},\overline{Y})$.
We note the second symmetry of $\mathcal{L}$ and $\sigma_3$:
\begin{equation}
\label{symm2}
\mathcal{L} = \mathcal{P} \overline{\mathcal{L}} \mathcal{P}, \quad
\sigma_3 = \mathcal{P} \sigma_3 \mathcal{P},
\end{equation}
where $\mathcal{P}$ is the parity transformation given by (\ref{operators}).
If $\lambda_0$ is an eigenvalue of the spectral problem~\eqref{original} with the eigenvector
$v_0 := (Y,Z)$, then so is $-\overline{\lambda}_0$ with the eigenvector
$\mathcal{P} \mathcal{T} v_0(x) = (\overline{Y(-x)},\overline{Z(-x)})$.
As a consequence of the two symmetries (\ref{symm1}) and (\ref{symm2}),
$-\lambda_0$ is also an eigenvalue with the eigenvector
$\mathcal{P} \sigma_1 v_0(x) = (Z(-x),Y(-x))$.
\end{proof}

Besides the spectral problem (\ref{original}), we also introduce the adjoint
spectral problem with the adjoint eigenvector denoted by $(Y^\#,Z^\#)$:
\small
\begin{equation}
        \arraycolsep=-1.7pt\def\arraystretch{1}
\mathcal{L}^* \left[\begin{array}{c} Y^\# \\ Z^\# \end{array}\right] = -i\lambda \sigma_3
\left[\begin{array}{c} Y^\# \\ Z^\# \end{array}\right],
\label{Adj}
\end{equation}
where
\begin{equation*}
\mathcal{L}^* := \left[\begin{array}{cc}
      -\partial_x^2 + V - i \gamma W - \mu - 2 g |\Phi|^2 & - g \Phi^2 \\
      - g \overline\Phi^2 & -\partial_x^2 + V + i \gamma W - \mu - 2 g |\Phi|^2
\end{array}\right].
\end{equation*}
\normalsize

\begin{remark}
Unless $\gamma = 0$ or $\Phi = 0$, the adjoint eigenvector $(Y^\#,Z^\#)$ cannot be
related to the eigenvector $(Y,Z)$ for the same eigenvalue $\lambda$.
\end{remark}

Our next assumption is on the existence of a nonzero isolated eigenvalue of the spectral problem (\ref{original}).

\begin{assumption}
\label{assumption-4}
We assume that there exists a simple isolated eigenvalue $\lambda_0 \in \mathbb{C} \backslash \{0\}$
of the spectral problems (\ref{original}) and (\ref{Adj})
with the eigenvector $v_0 := (Y,Z) \in D(L_0) \times D(L_0)$
and the adjoint eigenvector $v_0^\# := (Y^\#,Z^\#) \in D(L_0) \times D(L_0)$, respectively.
\end{assumption}

\begin{lemma}
\label{lem-PT-symmetry}
Under Assumption \ref{assumption-4}, if $\lambda_0\in i\mathbb{R}$, then
the corresponding eigenvectors $v_0 := (Y,Z)$ and $v_0^\# := (Y^\#,Z^\#)$
can be normalized to satisfy
\begin{equation}
\label{PT-sym-eigenvector}
Y(x) = \overline{Y(-x)}, \quad Z(x) = \overline{Z(-x)}, \quad  x \in \mathbb{R}
\end{equation}
and
\begin{equation}
\label{PT-sym-eigenvector-adjoint}
Y^\#(x) = \overline{Y^\#(-x)}, \quad Z^\#(x) = \overline{Z^\#(-x)}, \quad x \in \mathbb{R}.
\end{equation}
\end{lemma}

\begin{proof}
By Lemma \ref{lem-symmetry}, if $\lambda_0 \in i\mathbb{R}$ is a nonzero
eigenvalue with the eigenvector $v_0 := (Y,Z)$,
so is $-\overline{\lambda}_0 = \lambda_0$ with the eigenvector $\mathcal{P T} v_0$.
Since $\lambda_0$ is a simple eigenvalue, there is a constant $C \in \mathbb{C}$ such that
$v_0 = C \mathcal{P T} v_0$. Taking norms on both sides, we have $|C|=1$.
Therefore $C = e^{i\alpha}$ for some $\alpha\in [0,2\pi]$, and $\alpha$
can be chosen so that $v_0$ satisfy $v_0 = \mathcal{P T} v_0$ as in (\ref{PT-sym-eigenvector}).
The same argument applies to the adjoint eigenvector $v^\#_0 := (Y^\#,Z^\#)$.
\end{proof}

We shall now introduce the main object of our study, the Krein signature of the simple nonzero
isolated eigenvalue $\lambda_0$ in Assumption~\ref{assumption-4}.

\begin{definition}
\label{def-Krein}
The Krein signature of the eigenvalue $\lambda_0$ in Assumption \ref{assumption-4}
is the sign of the Krein quantity $K(\lambda_0)$ defined by
\begin{equation}
K(\lambda_0) = \langle v_0, \sigma_3 v_0^\# \rangle
    = \int_{\mathbb{R}} \left[ Y(x) \overline{Y^\#(x)} - Z(x)\overline{Z^\#(x)} \right] dx.
\label{pt-krein}
\end{equation}
\end{definition}

The following lemma states the main properties of the Krein quantity $K(\lambda_0)$.

\begin{lemma}
Assume (A4) and define $K(\lambda_0)$ by (\ref{pt-krein}). Then,
\begin{enumerate}
 \item $K(\lambda_0)$ is real if $\lambda_0 \in i\mathbb{R}  \backslash \{0\}$.
 \item $K(\lambda_0) \ne 0$ if $\lambda_0 \in i\mathbb{R}  \backslash \{0\}$.
 \item $K(\lambda_0) = 0$ if $\lambda_0 \in \mathbb{C}\backslash\{i\mathbb{R}\}$.
\end{enumerate}
\label{lem-krein}
\end{lemma}

\begin{proof}
First, we prove that if $f$ and $g$ are $\mathcal{PT}$-symmetric functions,
then their inner product $\langle f, g\rangle$ is real-valued. Indeed, this follows from
\begin{align*}
    \langle f,g\rangle &= \int_{\mathbb{R}} f(x)\overline{g(x)} dx =
        \int_0^{+\infty} \bigl( f(x)\overline{g(x)} + f(-x)\overline{g(-x)}\bigr)dx
        \\ &= \int_0^{+\infty}
        \bigl( f(x)\overline{g(x)} + \overline{f(x)}g(x) \bigr) dx.
\end{align*}
Since $\lambda_0 \in i \mathbb{R}  \backslash \{0\}$ is simple by Assumption~\ref{assumption-4},
then the eigenvectors $v_0 := (Y,Z)$ and $v_0^\# := (Y^\#,Z^\#)$ satisfy the $\mathcal{PT}$-symmetry
(\ref{PT-sym-eigenvector}) and (\ref{PT-sym-eigenvector-adjoint}) by
Lemma \ref{lem-PT-symmetry}. Hence, the inner products in the definition of $K(\lambda_0)$ in (\ref{pt-krein})
are real.

Next, we prove that $K(\lambda_0) \ne 0$ if $\lambda_0 \in i\mathbb{R}  \backslash \{0\}$ is simple.
Consider a generalized eigenvector problem for the spectral problem ~\eqref{original}:
\begin{equation}
           (\mathcal{L} + i \lambda_0 \sigma_3 ) \left[\begin{array}{c} Y_g \\ Z_g \end{array}\right]
           = \sigma_3 \left[\begin{array}{c} Y \\ Z \end{array}\right].
\label{eqn-gen}
\end{equation}
Since $\lambda_0 \notin \sigma_c(i\sigma_3\mathcal{L})$ is isolated and simple by Assumption \ref{assumption-4},
there exists a solution $v_g := (Y_g,Z_g) \in D(L_0) \times D(L_0)$ to the nonhomogeneous
equation~\eqref{eqn-gen} if and only if
$\sigma_3 v_0$ is orthogonal to $v_0^\#$, which is
the kernel of adjoint operator $\mathcal{L}^* + i\lambda_0\sigma_3$.
The orthogonality condition coincides with $K(\lambda_0) = 0$.
However, no $v_g$ exists since $\lambda_0 \in i \mathbb{R}  \backslash \{0\}$
is simple by Assumption~\ref{assumption-4}. Hence $K(\lambda_0) \neq 0$.

Finally, we show that $K(\lambda_0) = 0$ if $\lambda_0 \in \mathbb{C}\backslash\{i\mathbb{R}\}$.
Taking inner products for the spectral problems \eqref{original} and \eqref{Adj} with the corresponding
eigenvectors yields
\begin{equation*}
\begin{cases}
   \langle \mathcal{L} v_0, v_0^\# \rangle \!\!\!\! &= -i\lambda_0
   \langle\sigma_3 v_0, v_0^\#\rangle, \\
   \langle v_0, \mathcal{L}^* v_0^\# \rangle
           \!\!\!\! &= i\overline{\lambda}_0
   \langle v_0,   \sigma_3 v_0^\#   \rangle,
\end{cases}
\end{equation*}
hence
\begin{equation*}
i (\lambda_0 + \overline{\lambda}_0) K(\lambda_0) = 0.
\end{equation*}
If $\lambda_0 \in \mathbb{C}\backslash\{i\mathbb{R}\}$, then $\lambda_0 + \overline{\lambda}_0 \neq 0$ and $K(\lambda_0) = 0$.
\end{proof}

We shall now compare the Krein quantity $K(\lambda_0)$ in (\ref{pt-krein})
for simple eigenvalues of the $\mathcal{PT}$-symmetric spectral problem
(\ref{original}) with the corresponding definition of the Krein quantity
in the Hamiltonian case $\gamma = 0$ and in the linear $\mathcal{PT}$-symmetric case $\Phi = 0$.

In the Hamiltonian case ($\gamma = 0$), the operator $\mathcal{L}$ in the spectral problem (\ref{original})
is self-adjoint in $L^2(\mathbb{R})$, that is, $\mathcal{L} = \mathcal{L}^*$.
The standard definition of Krein quantity \cite{kapitulabook,mackay} is given by
\begin{equation}
\label{Krein-Ham}
\gamma = 0: \quad    K(\lambda_0) =
\left\langle \mathcal{L} v_0, v_0\right\rangle =
        -i\lambda_0 \int_\mathbb{R} \left[ |Y(x)|^2 - |Z(x)|^2 \right] dx.
\end{equation}
If $\gamma=0$ and $\lambda_0 \in i \mathbb{R}$, then the adjoint eigenvector $(Y^\#,Z^\#)$ satisfies
the same equation as $(Y,Z)$. Therefore, it is natural to choose the adjoint eigenvector in the form:
\begin{equation}
\label{Ham-case}
\gamma = 0: \quad Y^\#(x) = Y(x), \quad Z^\#(x) = Z(x), \quad x \in \mathbb{R},
\end{equation}
in which case the definition~\eqref{pt-krein} yields the integral
in the right-hand side of (\ref{Krein-Ham}). Note that the signs
of $K(\lambda_0)$ in (\ref{pt-krein}) and (\ref{Krein-Ham})
are the same if $\lambda_0 \in i \mathbb{R}_+$.

\begin{remark}
Since the potential $V$ is even in (\ref{potentials}),
the eigenvector $v_0 := (Y,Z)$ of the spectral problem (\ref{original})
for a simple eigenvalue $\lambda_0 \in i \mathbb{R} \backslash\{0\}$ is either even or odd
in the Hamiltonian case $\gamma=0$
by the parity symmetry. It follows from the $\mathcal{PT}$-symmetry (\ref{PT-sym-eigenvector})
that the $\mathcal{PT}$-normalized eigenvector $v_0$ is real if it is even and is purely imaginary if it is odd.
\label{remark-real-odd}
\end{remark}

\begin{remark}
Since the adjoint eigenvector $v_0^\# := (Y^\#,Z^\#)$ satisfying the
$\mathcal{PT}$-symmetry condition (\ref{PT-sym-eigenvector-adjoint})
is defined up to an arbitrary sign, the Krein quantity $K(\lambda_0)$ in (\ref{pt-krein}) is defined up to
the sign change. In the continuation of the NLSE (\ref{NLS})
with respect to the parameter $\gamma$ from the Hamiltonian case $\gamma = 0$,
the sign of the Krein quantity $K(\lambda_0)$ in (\ref{pt-krein}) can be chosen so that
it matches the sign of $K(\lambda_0)$ in (\ref{Krein-Ham}) for $\lambda_0 \in i \mathbb{R}_+$ and $\gamma = 0$.
In other words, the choice (\ref{Ham-case}) is always made for $\gamma = 0$ and the Krein
quantity $K(\lambda_0)$ is extended continuously with respect to the parameter $\gamma$.
\label{remark-cont}
\end{remark}

In the linear $\mathcal{PT}$-symmetric case ($\Phi = 0$), the spectral problem (\ref{original})
becomes diagonal. If $Z = 0$, then $Y$ satisfies the scalar Schr\"{o}dinger equation
\begin{equation}
\label{scalar-spectral}
    \left[-\partial^2_x + V(x) + i\gamma W(x) - \mu \right] Y(x) = -i\lambda Y(x).
\end{equation}
The $\mathcal{PT}$-Krein signature for the simple eigenvalue $\lambda_0 \in i \mathbb{R}$
of the scalar Schr\"{o}dinger equation (\ref{scalar-spectral}) is defined in
\cite{yang} as follows:
\begin{equation}
\label{Krein-Jang}
\Phi = 0, \quad Z = 0: \qquad    K(\lambda_0) = \int_\mathbb{R} Y(x)\overline{Y(-x)}dx.
\end{equation}
If $\lambda_0 \in i \mathbb{R}$, then the adjoint eigenfunction $Y^\#$ satisfies
a complex-conjugate equation to the spectral problem (\ref{scalar-spectral}),
which becomes identical to (\ref{scalar-spectral}) after the parity transformation.
Therefore, it is natural to choose the adjoint eigenfunction $Y^\#$ in the form:
\begin{equation*}
\Phi = 0, \quad Z = 0: \qquad  Y^\#(x) = Y(-x), \quad x \in \mathbb{R},
\end{equation*}
after which the definition~\eqref{pt-krein} with $Z = 0$ corresponds to the
definition (\ref{Krein-Jang}). If $Y=0$, then $Z$ satisfies the scalar Schr\"{o}dinger equation
\begin{equation}
\label{scalar-spectral-Z}
    \left[-\partial^2_x + V(x) - i\gamma W(x) - \mu \right] Z(x) = i\lambda Z(x).
\end{equation}
The $\mathcal{PT}$-Krein signature for the simple eigenvalue $\lambda_0 \in i \mathbb{R}$
of the scalar Schr\"{o}dinger equation (\ref{scalar-spectral-Z}) is defined by
\begin{equation}
\label{Krein-Z-Jang}
\Phi = 0, \quad Y = 0: \qquad    K(\lambda_0) = \int_\mathbb{R} Z(x)\overline{Z(-x)}dx,
\end{equation}
which coincides with the definition~\eqref{pt-krein} for  $Y = 0$ if the
adjoint eigenfunction $Z^\#$ is chosen in the form:
\begin{equation}
\label{Jang-Z}
\Phi = 0,\quad Y = 0: \qquad Z^\#(x) = -Z(-x), \quad x \in \mathbb{R}.
\end{equation}
Note that if the choice $Z^\#(x) = Z(-x)$ is made instead of (\ref{Jang-Z}),
then the definition (\ref{pt-krein}) with $Y = 0$
is negative with respect to the definition (\ref{Krein-Z-Jang}).

\section{Necessary conditions of the instability bifurcation}
\label{sec-proof}

Recall that the eigenvalue is called {\em semi-simple} if algebraic and geometric multiplicities coincide
and {\em defective} if algebraic multiplicity exceeds geometric multiplicity.
Here we consider the case when the nonzero eigenvalue $\lambda_0 \in i \mathbb{R}$ of the spectral problem
(\ref{original}) is defective with geometric multiplicity {\em one} and algebraic multiplicity {\em two}.
This situation occurs in the parameter continuations of the NLSE (\ref{NLS}) when two simple
isolated eigenvalues $\lambda_1, \lambda_2 \in i \mathbb{R} \backslash \{0\}$ coalesce at the point $\lambda_0 \neq 0$
and split into the complex plane resulting in the {\em instability bifurcation}.
We will use the parameter $\gamma$ to control the coalescence of two simple eigenvalues $\lambda_1, \lambda_2 \in i \mathbb{R}$.

Our main result states that the instability bifurcation occurs from the defective eigenvalue
$\lambda_0 \in i \mathbb{R}$ of algebraic multiplicity two
only if the Krein signatures of $K(\lambda_1)$ and $K(\lambda_2)$ for the two simple isolated eigenvalues $\lambda_1, \lambda_2 \in i \mathbb{R}$
before coalescence are opposite to each other. Therefore, we obtain the necessary condition
for the instability bifurcation in the $\mathcal{PT}$-symmetric spectral problem (\ref{original}),
which has been proven for the Hamiltonian spectral problems~\cite{kapitulabook,mackay}.

\begin{remark}
The necessary condition for instability bifurcation allows us to predict 
the transition from stability to instability when a pair of imaginary eigenvalues collide.
Pairs with the same Krein signature do not bifurcate off the imaginary axis if they collide. 
In the contrast, pairs with the opposite Krein signature may bifurcate off the imaginary axis 
under a technical non-degeneracy condition (\ref{non-degeneracy}) below. 
\end{remark}

First, we state why the perturbation theory can be applied to the spectral problem (\ref{original}).

\begin{lemma}
\label{lem-operator-L}
Under Assumptions \ref{assumption-1}, \ref{assumption-2}, and \ref{assumption-3},
the operator
\begin{equation*}
\mathcal{L} : D(L_0) \times D(L_0) \to L^2(\mathbb{R}) \times L^2(\mathbb{R})
\end{equation*}
in the spectral problem (\ref{original}) is real-analytic with respect to
$\gamma \in (-\gamma_*,\gamma_*)$.
Consequently, if $\mathcal{L}(\gamma_0)$ with $\gamma_0 \in (-\gamma_*,\gamma_*)$
has a spectrum consisting of two separated parts, then the subspaces of $L^2(\mathbb{R}) \times L^2(\mathbb{R})$
corresponding to the separated parts are also real-analytic in $\gamma$.
\end{lemma}

\begin{proof}
Operator $\mathcal{L}$ depends on $\gamma$ via the potential $i \gamma W$ and the bound state $\Phi$,
the latter is real-analytic for $\gamma \in (-\gamma_*,\gamma_*)$ by Assumption~\ref{assumption-3}.
The assertion of the lemma follows from Theorem 1.7 in Chapter VII on p. 368 in \cite{Kato}.
\end{proof}

By Lemma \ref{lem-operator-L}, simple isolated eigenvalues
$\lambda_1, \lambda_2 \in i \mathbb{R}$  of the spectral problem (\ref{original})
and their eigenvectors $v_1 := (Y_1,Z_1)$ and $v_2 := (Y_2,Z_2)$ are continued analytically
in $\gamma$ before the coalescence point. Similarly,
the adjoint eigenvectors $v_1^\# := (Y_1^\#,Z_1^\#)$ and $v_2^\# := (Y_2^\#,Z_2^\#)$ of the adjoint spectral
problem (\ref{Adj}) for $\lambda_1, \lambda_2 \in i \mathbb{R}$ are continued analytically in $\gamma$. Therefore, the
Krein quantities $K(\lambda_1)$ and $K(\lambda_2)$ are continued analytically in $\gamma$.

Let $\gamma_0$ denote the bifurcation point when the two eigenvalues coalesce:
$\lambda_1 = \lambda_2 = \lambda_0 \in i \mathbb{R} \backslash \{0\}$.
For this $\gamma_0 \in \mathbb{R}$, we can define a small parameter $\varepsilon\in\mathbb{R}$ such that
$\gamma = \gamma_0 + \varepsilon$. If $\mathcal{L}$ is denoted by $\mathcal{L}(\gamma)$,
then $\mathcal{L}(\gamma)$ can be represented by the Taylor expansion:
\begin{equation}
\mathcal{L}(\gamma) = \mathcal{L}(\gamma_0)
	+ \varepsilon \mathcal{L'}(\gamma_0) + \varepsilon^2 \hat{\mathcal{L}}(\varepsilon),
\label{operator:series}
\end{equation}
where $\hat{\mathcal{L}}(\varepsilon)$ denotes the remainder terms,
\begin{equation}
\mathcal{L}'(\gamma_0) =
\left[\begin{array}{cc}
        iW - 2 g \partial_\gamma |\Phi(\gamma_0)|^2 & - g \partial_\gamma \Phi^2(\gamma_0) \\[2pt]
        -g \partial_\gamma \overline{\Phi^2(\gamma_0)} & -iW - 2 g \partial_\gamma |\Phi(\gamma_0)|^2
\end{array}\right],
\label{operator:lprime}
\end{equation}
and $\partial_\gamma$ denotes a partial derivative with respect to the parameter $\gamma$.
Since the remainder terms in $\hat{\mathcal{L}}(\varepsilon)$ come from the second derivative of $\Phi$ in $\gamma$ near $\gamma_0$,
then $\hat{\mathcal{L}}(\varepsilon) \in L^2(\mathbb{R}) \cap L^{\infty}(\mathbb{R})$ thanks to Assumption \ref{assumption-3}.

Instead of Assumption~\ref{assumption-4}, we shall now use the following assumption.

\setcounter{assumption}{3}
\renewcommand{\theassumption}{(A\arabic{assumption}$^\prime$)}
\begin{assumption}
\label{assumption-5}
For $\gamma = \gamma_0$, we assume that there exists a
defective isolated eigenvalue $\lambda_0 \in i \mathbb{R} \backslash \{0\}$
of the spectral problems (\ref{original}) and (\ref{Adj})
with the eigenvector $v_0 := (Y_0,Z_0) \in D(L_0) \times D(L_0)$,
the generalized eigenvector $v_g := (Y_g,Z_g) \in D(L_0) \times D(L_0)$
and the adjoint eigenvector $v_0^\# := (Y^\#_0,Z^\#_0) \in D(L_0) \times D(L_0)$,
the adjoint generalized eigenvector $v_g^\# := (Y^\#_g,Z^\#_g) \in D(L_0) \times D(L_0)$,
respectively.
\end{assumption}
\renewcommand{\theassumption}{(A\arabic{assumption})}

By setting $\lambda_0 = i \Omega_0$, we can write the linear equations
for the eigenvectors and generalized eigenvectors in  Assumption~\ref{assumption-5}:
\begin{gather}
\mathcal{L}(\gamma_0)v_0 = \Omega_0 \sigma_3 v_0, \qquad
\mathcal{L}(\gamma_0)v_g = \Omega_0 \sigma_3 v_g + \sigma_3 v_0,
\label{eig:omega:orig} \\
\mathcal{L}^*(\gamma_0)v_0^\# =
    \Omega_0 \sigma_3 v_0^\#, \qquad
\mathcal{L}^*(\gamma_0)v_g^\# = \Omega_0 \sigma_3 v_g^\# + \sigma_3 v_0^\#.
\label{eig:omega:adj}
\end{gather}
The solvability conditions for the inhomogeneous equations in (\ref{eig:omega:orig}) and (\ref{eig:omega:adj})
yield the following elementary facts.

\begin{lemma}
\label{lem-facts}
Under Assumption~\ref{assumption-5}, we have
\begin{equation}
\label{fact-1}
K(\lambda_0) = \langle v_0,\sigma_3 v_0^\#\rangle = 0.
\end{equation}
and
\begin{equation}
\label{fact-2}
\langle v_g,\sigma_3 v_0^\#\rangle = \langle v_0,\sigma_3 v_g^\#\rangle \neq 0.
\end{equation}
\end{lemma}

\begin{proof}
Since $v_g$ exists by Assumption~\ref{assumption-5}, the
solvability condition for (\ref{eig:omega:orig}) implies (\ref{fact-1}),
see similar computations in Lemma \ref{lem-krein}.
Since the eigenvalue $\lambda_0$ is double, no
second generalized eigenvector $\tilde{v}_g$ exists such that
\begin{equation*}
\mathcal{L}(\gamma_0) \tilde{v}_g = \Omega_0 \sigma_3 \tilde{v}_g + \sigma_3 v_g.
\end{equation*}
The nonsolvability condition for this equation implies $\langle v_g, \sigma_3 v_0^\# \rangle \ne 0$.
Finally, equations ~\eqref{eig:omega:orig} and \eqref{eig:omega:adj} yield
\begin{align*}
    \langle v_g, \sigma_3 v_0^\#\rangle &=
    \langle v_g, (\mathcal{L}^*-\Omega_0 \sigma_3)v_g^\#\rangle =
    \langle (\mathcal{L}-\Omega_0 \sigma_3)v_g, v_g^\#\rangle \\ &=
    \langle \sigma_3 v_0, v_g^\#\rangle =
    \langle v_0, \sigma_3 v_g^\#\rangle,
\end{align*}
which proves the symmetry in (\ref{fact-2}).
\end{proof}

\begin{remark}
Since the generalized eigenvectors are given by solutions of the inhomogeneous linear equations
(\ref{eig:omega:orig}) and (\ref{eig:omega:adj}) and the eigenvectors satisfy
the $\mathcal{PT}$-symmetry (\ref{PT-sym-eigenvector}) and
(\ref{PT-sym-eigenvector-adjoint}), the generalized eigenvectors also satisfy the same $\mathcal{PT}$-symmetry
(\ref{PT-sym-eigenvector}) and (\ref{PT-sym-eigenvector-adjoint}).
\end{remark}

The following result gives the necessary condition that the defective eigenvalue $\lambda_0$ in Assumption~\ref{assumption-5}
splits into the complex plane in a one-sided neighborhood of the bifurcation point $\gamma_0$.

\begin{theorem}
\label{theorem-main}
Assume \ref{assumption-1}, \ref{assumption-2}, \ref{assumption-3}, \ref{assumption-5}, and
the non-degeneracy condition
\begin{equation}
\label{non-degeneracy}
\langle \mathcal{L'}(\gamma_0) v_0, v_0^\#\rangle \ne 0.
\end{equation}
There exists $\varepsilon_0 > 0$ such that two simple eigenvalues $\lambda_1,\lambda_2$
of the spectral problem (\ref{original}) exist near $\lambda_0$ for every $\varepsilon \in (-\varepsilon_0,\varepsilon_0) \backslash \{0\}$
with $\lambda_{1,2} \to \lambda_0$ as $\varepsilon \to 0$.
On one side of $\varepsilon = 0$, the eigenvalues are $\lambda_1, \lambda_2 \in i \mathbb{R}$ and
\begin{equation}
\label{opposite-Krein}
{\rm sign} K(\lambda_1) = -{\rm sign} K(\lambda_2).
\end{equation}
On the other side of $\varepsilon = 0$, the eigenvalues are $\lambda_1,\lambda_2 \notin i \mathbb{R}$.
\end{theorem}

\begin{proof}
We are looking for an eigenvalue $\Omega(\varepsilon)$ of the perturbed spectral problem
\begin{equation}
\left[ \mathcal{L}_0
	+ \varepsilon \widetilde{\mathcal{L}}(\varepsilon) \right] v(\varepsilon) =
\Omega(\varepsilon) \sigma_3 v(\varepsilon),
\label{original:perturbed}
\end{equation}
such that $\Omega(\varepsilon) \to \Omega_0$ as $\varepsilon \to 0$.
Here we denote operators from the decomposition \eqref{operator:series} as
$\mathcal{L}_0 = \mathcal{L}(\gamma_0)$ and $\widetilde{\mathcal{L}}(\varepsilon) =
\mathcal{L}'(\gamma_0) + \varepsilon \hat{\mathcal{L}}(\varepsilon)$.
Since $\Omega_0$ is a defective eigenvalue of geometric multiplicity {\em one}
and algebraic multiplicity {\em two}, we apply Puiseux expansions~\cite{Knopp}:
\begin{equation}
\left\{ \begin{array}{l}
        \Omega(\varepsilon) = \Omega_0 + \varepsilon^{1/2} \widetilde\Omega(\varepsilon), \\
        v(\varepsilon) = v_0 + \varepsilon^{1/2} a(\varepsilon) v_g
        + \varepsilon \widetilde{v_1}(\varepsilon),
                    \end{array} \right.
                    \label{Puiseux}
\end{equation}
where $v_0$ and $v_g$ are the eigenvector and the generalized
eigenvector of the spectral problem (\ref{eig:omega:orig}),
$a(\varepsilon)$ is the projection coefficient to be defined, and
$\widetilde\Omega(\varepsilon)$ and $\widetilde{v_1}(\varepsilon)$ are the remainder terms.
To define $\widetilde{v_1}(\varepsilon)$ uniquely, we add the orthogonality condition
\begin{equation}
    \langle \widetilde{v_1}(\varepsilon), \sigma_3 v_0^\# \rangle
= \langle \widetilde{v_1}(\varepsilon), \sigma_3 v_g^\# \rangle = 0.
\label{ve-orthogonal}
\end{equation}
Plugging~\eqref{Puiseux} into~\eqref{original:perturbed} and dropping the dependence on $\varepsilon$
for $\widetilde{\mathcal{L}}$, $\widetilde{v}_1$, $a$ and $\widetilde{\Omega}$ gives us
the nonhomogeneous equation
\begin{gather}
\left( \mathcal{L}_0 - \Omega_0 \sigma_3
+ \varepsilon \widetilde{\mathcal{L}} - \varepsilon^{1/2} \widetilde{\Omega} \sigma_3 \right) \widetilde{v}_1 = 
\varepsilon^{-1/2} (\widetilde{\Omega} - a) \sigma_3 v_0
    - \widetilde{\mathcal{L}} v_0 + a
\bigl( \widetilde{\Omega} \sigma_3
    - \varepsilon^{1/2} \widetilde{\mathcal{L}} \bigr) v_g.
\label{ve-equation}
\end{gather}
By Assumption~\ref{assumption-5}, the limiting operator $\sigma_3 (\mathcal{L}_0 - \Omega_0 \sigma_3)$ has the two-dimensional
generalized null space $X_0 = {\rm span}\{v_0,v_g\} \subset  L^2(\mathbb{R}) \times L^2(\mathbb{R})$.
Since $\Omega_0 \notin \sigma_c(\sigma_3 \mathcal{L}_0)$ is isolated from the rest of the spectrum of $\sigma_3 \mathcal{L}_0$,
the range of $\sigma_3 (\mathcal{L}_0 - \Omega_0 \sigma_3)$ is orthogonal with respect to generalized null space
$Y_0 = {\rm span}\{ \sigma_3 v_0^\#, \sigma_3 v_g^\#\} \subset L^2(\mathbb{R}) \times L^2(\mathbb{R})$
of the adjoint operator $(\mathcal{L}_0^* - \Omega_0 \sigma_3) \sigma_3$. As a result, $\sigma_3 (\mathcal{L}_0 - \Omega_0 \sigma_3)$ is
invertible on an element of $Y_0^\perp$ and the inverse operator
is uniquely defined and bounded in $Y_0^{\perp}$. In other words, 
there exist positive constants $\varepsilon_0$, $\Omega_0$, and $C_0$ such that for
all $|\varepsilon| \leq \varepsilon_0$, $|\widetilde{\Omega}| \leq \Omega_0$, and all
$\sigma_3 f\in Y_0^\perp$, there exists a unique
$(\mathcal{L}_0 - \Omega_0 \sigma_3)^{-1} f \in D(L_0) \times D(L_0)$
satisfying the orthogonality conditions (\ref{ve-orthogonal}) and the bound
\begin{equation}
\label{resolvent-estimate}
    \| (\mathcal{L}_0 - \Omega_0 \sigma_3)^{-1} f\|_{L^2} \le C_0 \|f\|_{L^2}.
\end{equation}
In order to provide existence of a unique $(\mathcal{L}_0 - \Omega_0 \sigma_3)^{-1} f$,
we add the orthogonality constraints $\langle f,v_0^\#\rangle = \langle f,v_g^\#\rangle = 0$.
By using~\eqref{fact-2} and ~\eqref{ve-orthogonal}, we obtain two equations:
\begin{gather}
    \varepsilon \langle \widetilde{\mathcal{L}} \widetilde{v}_1, v_0^\#\rangle +
    \langle \widetilde{\mathcal{L}} v_0, v_0^\# \rangle =
    \widetilde\Omega a \langle v_g, \sigma_3 v_0^\#\rangle
    - \varepsilon^{1/2} a \langle \widetilde{\mathcal{L}} v_g, v_0^\# \rangle,
\label{proj-v0s}
\end{gather}
and
\begin{align}
    \varepsilon \langle \widetilde{\mathcal{L}} \widetilde{v_1}, v_g^\#\rangle +
    \langle \widetilde{\mathcal{L}} v_0,v_g^\#\rangle =
    & \, \widetilde\Omega a \langle v_g,\sigma_3 v_g^\#\rangle \notag \\
    &+ \varepsilon^{-1/2} (\widetilde\Omega - a) \langle v_0,\sigma_3 v_g^\#\rangle
    - \varepsilon^{1/2} a \langle \widetilde{\mathcal{L}} v_g,v_g^\#\rangle.
\label{proj-v1s}
\end{align}
Since $\widetilde{\mathcal{L}}$ and $\widetilde\Omega \sigma_3$
are relatively compact perturbations to $(\mathcal{L}_0 - \Omega_0 \sigma_3)$,
under the constraints~\eqref{proj-v0s} and \eqref{proj-v1s}, there exists a unique
solution of the nonhomogeneous equation~\eqref{ve-equation}
satisfying the orthogonality conditions (\ref{ve-orthogonal})
and the resolvent estimate~\eqref{resolvent-estimate}. In particular, there exist positive constants
$\varepsilon_0$, $\Omega_0$, $A_0$, and $C_0$ such that for
all $|\varepsilon| \leq \varepsilon_0$, $|\widetilde{\Omega}| \leq \Omega_0$, and $|a| \leq A_0$,
the solution $\widetilde{v}_1 \in D(L_0) \times D(L_0)$ of equation (\ref{ve-equation}) satisfies the estimate
\begin{equation}
    \| \widetilde{v_1} \|_{L^2} \le C_0 \left( \varepsilon^{-1/2}|a - \widetilde{\Omega}| + 1 + |\widetilde{\Omega}|^2\right).
\label{ve-estimate}
\end{equation}
Equation~\eqref{proj-v1s} yields
\begin{align*}
\varepsilon^{-1/2} (a - \widetilde{\Omega}) =
\frac{1}{\langle v_0, \sigma_3 v_g^\#\rangle}
&\left( \widetilde\Omega a \langle v_g, \sigma_3 v_g^\#\rangle
            - \varepsilon^{1/2} a \langle \widetilde{\mathcal{L}} v_g,v_g^\#\rangle\right. \notag \\
            &\quad - \left.\langle \widetilde{\mathcal{L}} v_0, v_g^\#\rangle
        - \varepsilon \langle \widetilde{\mathcal{L}} \widetilde{v}_1, v_g^\#\rangle \right),
\end{align*}
where~$\langle v_0, \sigma_3 v_g^\#\rangle\ne 0$ due to Lemma~\ref{lem-facts}.
Combining with the estimate (\ref{ve-estimate}), we obtain for some $C_1>0$
\begin{equation}
\label{ve-estimate-a}
    |a - \widetilde\Omega | \le C_1 \varepsilon^{1/2} (1 + |\widetilde{\Omega}|^2) \quad \text{and} \quad
    \|\tilde{v}_1\|_{L^2} \le C_1 (1 + |\widetilde{\Omega}|^2).
\end{equation}
Equation~\eqref{proj-v0s} yields
\begin{gather*}
    \widetilde\Omega a = \frac{1}{\langle v_g, \sigma_3 v_0^\#\rangle}
\left( \langle \widetilde{\mathcal{L}} v_0, v_0^\#\rangle
    + \varepsilon^{1/2} a \langle \widetilde{\mathcal{L}} v_g,v_0^\#\rangle
+ \varepsilon \langle \widetilde{\mathcal{L}} \widetilde{v_1}, v_0^\# \rangle \right),
\end{gather*}
where $\langle v_g, \sigma_3 v_0^\#\rangle\ne 0$ due to Lemma~\ref{lem-facts}.
Thanks to (\ref{ve-estimate-a}), we obtain
\begin{equation*}
    | \widetilde\Omega - \Omega_g | \le C_2 \varepsilon^{1/2},
\end{equation*}
where $C_2>0$ is a constant, and $\Omega_g$ is a root of the quadratic equation
\begin{equation}
    \Omega_g^2 = \frac{\langle \mathcal{L}'(\gamma_0) v_0,v_0^\#\rangle }
    {\langle v_g,\sigma_3 v_0^\#\rangle},
                      \label{eq:omegag}
\end{equation}
with $\mathcal{L}'(\gamma_0)$ given by~\eqref{operator:lprime}.
Since $\mathcal{L}'(\gamma_0) v_0$, $v_g$, and $v_0^\#$ satisfy the
$\mathcal{PT}$-conditions (\ref{PT-sym}), (\ref{PT-sym-eigenvector}), and (\ref{PT-sym-eigenvector-adjoint}),
both the nominator and the denominator of (\ref{eq:omegag}) are real-valued by the same computations as in
the proof of Lemma \ref{lem-krein}. By the assumption (\ref{non-degeneracy}),
$\Omega_g^2$ is nonzero, either positive or negative.

Let us assume that $\Omega_g^2 > 0$ without loss of generality and pick $\Omega_g > 0$.
Then $\varepsilon^{1/2} \Omega_g \in \mathbb{R}$ if $\varepsilon > 0$ and we obtain the expansions
for the two simple eigenvalues:
\begin{equation*}
\begin{cases}
    \Omega_1(\varepsilon) = \Omega_0 + \varepsilon^{1/2} \Omega_g
        + \mathcal{O}(\varepsilon), \\
    \Omega_2(\varepsilon) = \Omega_0 - \varepsilon^{1/2} \Omega_g
        + \mathcal{O}(\varepsilon)
\end{cases}
\end{equation*}
and their corresponding eigenvectors:
\begin{equation*}
\begin{cases}
    v_1(\varepsilon) = v_0 + \varepsilon^{1/2} \Omega_g v_g
        + \mathcal{O}(\varepsilon), \\
    v_2(\varepsilon) = v_0 - \varepsilon^{1/2} \Omega_g v_g
        + \mathcal{O}(\varepsilon).
\end{cases}
\end{equation*}
The same expansions hold for eigenvectors of the adjoint spectral problems
corresponding to the same eigenvalues $\Omega_1,\Omega_2$:
\begin{equation*}
\begin{cases}
    v_1^\#(\varepsilon) = v_0^\# + \varepsilon^{1/2} \Omega_g v_g^\#
        + \mathcal{O}(\varepsilon), \\
    v_2^\#(\varepsilon) = v_0^\# - \varepsilon^{1/2} \Omega_g v_g^\#
        + \mathcal{O}(\varepsilon).
\end{cases}
\end{equation*}
The leading order of Krein quantitites for eigenvalues $\lambda_1 = i\Omega_1$ and $\lambda_2 = i\Omega_2$
is given by
\begin{equation*}
\begin{cases}
    K(\lambda_1) = \langle v_1,\sigma_3 v_1^\#\rangle
    = \varepsilon^{1/2} \Omega_g
    \langle v_g, \sigma_3 v_0^\#\rangle + \overline{\varepsilon^{1/2} \Omega_g}
        \langle v_0, \sigma_3 v_g^\#\rangle + \mathcal{O}(\varepsilon),
        \\
        K(\lambda_2) = \langle v_2,\sigma_3 v_2^\#\rangle =
        -\varepsilon^{1/2} \Omega_g \langle v_g,\sigma_3 v_0^\#\rangle
    - \overline{\varepsilon^{1/2} \Omega_g} \langle v_0,\sigma_3 v_g^\#\rangle
    + \mathcal{O}(\varepsilon),
\end{cases}
\end{equation*}
which is simplified with the help of (\ref{fact-2}) to
\begin{equation*}
\begin{cases}
    K(\lambda_1) = 2\varepsilon^{1/2} \Omega_g \langle v_g, \sigma_3 v_0^\#\rangle
            + \mathcal{O}(\varepsilon), \\
    K(\lambda_2) = -2\varepsilon^{1/2} \Omega_g \langle v_g,\sigma_3 v_0^\#\rangle
                + \mathcal{O}(\varepsilon).

\end{cases}
\end{equation*}
Since $\epsilon^{1/2} \Omega_g \in \mathbb{R}$ and $\langle v_g, \sigma_3 v_0^\# \rangle \neq 0$,
we obtain (\ref{opposite-Krein}).
If $\varepsilon < 0$, then $\epsilon^{1/2} \Omega_g \in i \mathbb{R}$,
so that $\lambda_1, \lambda_2 \notin i \mathbb{R}$.
\end{proof}

\begin{remark}
If the non-degeneracy assumption (\ref{non-degeneracy}) is not satisfied, then
$\Omega_g = 0$ and the perturbation theory must be extended to the next order. In this case,
the defective eigenvalue $\lambda_0 = i \Omega_0$ may split along $i \mathbb{R}$
both for $\varepsilon > 0$ and $\varepsilon < 0$.
\end{remark}

\section{Numerical Approximations}
\label{sec-numerics}

We approximate nonlinear modes $\Phi$ of the stationary NLSE~\eqref{NLSstat} and eigenvectors
$(Y,Z)$ of the spectral problem~\eqref{original} with the Chebyshev interpolation method~\cite{trefethen}.
This method was recently applied to massive Dirac equations in~\cite{yusuke}.
Chebyshev polynomials are defined on the interval $[-1,1]$. The stationary NLSE~\eqref{NLSstat} is defined
on the real line, therefore we make a coordinate transformation for the Chebyshev grid points
$\{ z_j = \cos(\frac{j\pi}{N})\}_{j=0}^{j=N}$:
\begin{equation}
x_j = L \arctanh(z_j), \quad j = 1,2,\ldots,N-1,
\label{eq:transform}
\end{equation}
where $x_0 = +\infty$ and $x_N = -\infty$.
The scaling parameter $L$ is chosen so that the grid points $\{ x_j \}_{j =1}^{j = N-1}$ are 
concentrated in the region where the nonlinear mode $\Phi$ changes fast. 
We apply the chain rule for the second derivative:
\begin{gather*}
\frac{d^2 u}{dx^2} = \frac{d}{dx} \left(\frac{du}{dx}\right) =
\frac{d}{dz} \left(\frac{du}{dz} \frac{dz}{dx}\right)=
\frac{d^2 u}{dz^2} \left(\frac{dz}{dx}\right)^2 + \frac{du}{dz} \frac{d^2 z}{dx^2},
\end{gather*}
where
\begin{eqnarray*}
\frac{dz}{dx} = \frac{1}{L} \sech^2 \left(\frac{x}{L}\right) = \frac{1}{L} (1-z^2)
\end{eqnarray*}
and
\begin{eqnarray*}
\frac{d^2 z}{dx^2} = -\frac{2}{L^2} \sech^2\left(\frac{x}{L}\right)
\tanh\left(\frac{x}{L}\right) = -\frac{2}{L^2} z(1-z^2).
\end{eqnarray*}
The first and second derivatives for $\partial_z$ and $\partial^2_z$ are approximated by
the Chebyshev differentiation matrices $D_N$ and $D_N^2$, respectively (see~\cite{trefethen}, p.53).

The stationary NLSE~\eqref{NLSstat} is written in the form:
\begin{equation}
\label{root-finding}
F(\Phi) := (- \partial_x^2 + V + i\gamma W - \mu - g|\Phi|^2) \Phi = 0.
\end{equation}
We fix $\mu$, $\gamma$, $g$, $V(x)$, $W(x)$ and use Newton's method to look for a solution
$\Phi$ satisfying Assumption \ref{assumption-3}:
\begin{equation}
\left[\begin{array}{c} \Phi_{n+1} \\ \bar\Phi_{n+1} \end{array}\right] =
\left[\begin{array}{c} \Phi_n \\ \bar\Phi_n \end{array}\right]
- \mathcal{L}^{-1}_n \left[\begin{array}{c} F(\Phi_n) \\ \bar F(\Phi_n) \end{array}\right],
\label{newton}
\end{equation}
where $\mathcal{L}_n$ is the Jacobian operator to the nonlinear problem (\ref{root-finding}),
which coincides with~\eqref{jacobian} computed at $\Phi_n$.
Since $\Phi(x_0) = \Phi(x_N) = 0$,
the Jacobian operator $\mathcal{L}_n$ is represented by the $2(N-1) \times 2(N-1)$ matrix.

It follows by the gauge transformation that
\begin{equation*}
\mathcal{L} \left[ \begin{array}{c} i\Phi \\ -i\bar\Phi \end{array} \right]
= \left[ \begin{array}{c} 0 \\ 0 \end{array} \right],
\end{equation*}
where $\mathcal{L}$ is given by (\ref{jacobian}). 
Therefore, $\mathcal{L}$ is a singular operator for every parameter choice of equation~\eqref{root-finding}.
However, if the eigenvector satisfies the symmetry $\bar{Z} = Y$, then the eigenvector does not 
satisfy the $\mathcal{PT}$-symmetry:
\begin{equation*}
\mathcal{PT}\left[ \begin{array}{c} i\Phi \\ -i\bar\Phi \end{array} \right]
= \left[ \begin{array}{c} -i\overline{\Phi(-x)} \\ i\Phi(-x) \end{array} \right]
= - \left[ \begin{array}{c} i\Phi \\ -i\bar\Phi \end{array} \right].
\end{equation*}
Hence, $\mathcal{L}$ is invertible on the space of $\cal{PT}$-symmetric functions satisfying (\ref{PT-sym}).
In terms of the coefficients of Chebyshev polynomials, the restriction means that
the even-numbered coefficients are purely real, whereas the odd-numbered
coefficients are purely imaginary.

Choosing a first guess for the iterative procedure~\eqref{newton} depends on the choice of
the potentials $V$ and $W$. For the Scarf II potential~\eqref{pot-Wadati}, one can use a scalar multiple of
the $\sech(x)$ function for the first branch of solutions and a scalar multiple of
the $\sech(x)\tanh(x)$ function for the second branch of solutions~\cite{ahmed}. For the confining potential~\eqref{pot-BEC},
one can use the corresponding Gauss-Hermite functions of the linear system for each branch~\cite{zezyulin}.

The spectral problem~\eqref{original} uses the same operator $\mathcal{L}$ and can
be discretized similarly. One looks for eigenvalues and eigenvectors of the discretized
matrix by using the standard numerical methods for non-Hermitian matrices.
For example, MATLAB$^\text{\textregistered}$ performs these computations by using the QZ algorithm.

Throughout the numerical results, we pick the value of a scaling parameter $L$ to be $L=10$.
This choice ensures that $\Phi$ remains nonzero up to $16$ decimals on
the interior grid points $\{ x_j \}_{j=1}^{j=N-1}$. The algorithm was tested on the exact
solution derived in~\cite{wadati} for the Scarf II potential~\eqref{pot-Wadati} with $V_0=1$ 
and $\mu=\gamma=-1$:
\begin{equation}
\label{exact-sol}
    \Phi_{exact}(x) = \sin \alpha \sech(x) \exp\left( \frac{i}{2} \cos\alpha \arctan(\sinh(x)) \right),
\end{equation}
where $\alpha = \arccos(2/3)$.  Table~\ref{err-table} shows a good agreement between exact and numerical
results.

\begin{table}
    \center
    \begin{tabular}{cc}
        \hline
	& $\|\Phi_{exact} - \Phi_{numerical} \|^2$ \\
        \hline
        N = 50 &  $1.5 \times 10^{-6}$ \\
        \hline
        N = 100 &  $2.4 \times 10^{-13}$ \\
        \hline
        N = 500 & $2.2\times 10^{-13}$ \\
        \hline
\end{tabular}
\caption{The numerical error for the exact solution (\ref{exact-sol})
versus $N$.}
\label{err-table}
\end{table}

Once we computed eigenvalues and eigenvectors for the spectral problem~\eqref{original},
we proceed to computations of the Krein quantity defined by ~\eqref{pt-krein}.
Several obstacles arise in the definition of the Krein quantity:
\begin{enumerate}
    \item Eigenvectors of the Chebyshev discretization matrices are normalized with respect to $z$.
    \item Eigenvectors are not necessarily $\cal{PT}$-symmetric.
    \item The sign of the adjoint eigenvectors relative to the eigenvectors is undefined.
\end{enumerate}
Here we explain how to deal with these technical difficulties.
\begin{enumerate}
\item The eigenvectors are normalized in the $L^2([-1,1])$ norm with respect to the variable $z$.
In order to normalize them in the $L^2(\mathbb{R})$ norm with respect to the variable $x$,
we perform the change of coordinates~\eqref{eq:transform}. In particular, we use integration with
the composite trapezoid method on the grid points $\{ x_j \}_{j=1}^{j=N-1}$ and neglect integrals
for $(-\infty,x_{N-1})$ and $(x_1,+\infty)$.

\item In order to restore the $\mathcal{PT}$-symmetry condition~\eqref{PT-sym-eigenvector},
we multiply the component $Y$ returned from the eigenvector by $e^{i \theta}$ with $\theta \in [0,2\pi]$
and require
\begin{gather*}
    e^{i\theta} Y(x) = e^{-i\theta} \overline{Y(-x)} \quad \Rightarrow \quad
    2i\theta = \log\frac{\overline{Y(-x)}}{Y(x)},
\end{gather*}
where the point $x$ is chosen so that $Y(x)$ and $Y(-x)$ are nonzero.
For example, we compute $\theta$ for all interior grid points $\{ x_j \}_{j=1}^{j=N-1}$
for which $Y(x_j) \neq 0$ and take the average. Both $Y$ and $Z$ in the same eigenvector
are rotated with the same angle $\theta$. Similarly,
this step is performed for $Y^\#$ and $Z^\#$ according to the
$\mathcal{PT}$-symmetry condition~\eqref{PT-sym-eigenvector-adjoint}.

\item We fix the sign of the adjoint eigenvectors at the Hamiltonian case $\gamma = 0$
by using (\ref{Ham-case}). Then we continue the eigenvectors and the adjoint
eigenvectors for simple eigenvalues before coalescence points.
Numerically, we take two steps in $\gamma$: $\gamma_1 < \gamma_2$,
with $|\gamma_2 - \gamma_1| \ll 1$. Suppose that the sign of eigenvector for $\gamma_1$ has been
chosen already. We take eigenvectors for $\gamma_1$ and $\gamma_2$ and compare them.
If eigenvectors have been made $\mathcal{PT}$-symmetric and properly normalized,
then the norm of their difference is either small
(the eigenvectors are almost the same) or close to $2$ (the eigenvectors are negatives of each other). We choose
the sign of the eigenvector so that the norm of their difference is small.
\end{enumerate}

With the refinements described above, we can now compute the Krein quantity $K(\lambda)$ defined by
(\ref{pt-krein}) using the same numerical method as the one used for computing the norms of eigenvectors.

In numerical computations, we have often encountered situations when eigenvalues nearly coalesce,
but the standard MATLAB$^\text{\textregistered}$ numerical routines do not approximate well
the coalescence of eigenvalues. In order to check if the eigenvectors
are linearly dependent near the possible coalescence point, we compute the norm
of the difference between the two eigenvectors (or opposites of each other)
for the two simple eigenvalues and plot it with respect to the parameter $\gamma$.
If the difference between the two eigenvectors vanishes as $\gamma$ is increased towards the
coalescence point, we say that the defective eigenvalue arises at the bifurcation
point. If the difference remains finite,
either we are dealing with the semi-simple eigenvalue at the coalescence point or
the two simple eigenvalues pass each other without coalescence.

\section{Numerical Examples}
\label{sec-examples}

In the numerical examples, we set $N=500$. This gives enough
accuracy for computing eigenvalues, as it was shown in~\cite{yusuke}.
We will demonstrate numerical results on Figures~\ref{fig:yang},\ref{fig:wadati},\ref{fig:kev3}
and~\ref{fig:kev4}. Each figure displays branches of the nonlinear modes $\Phi$
versus a parameter used in the numerical continuations
(either $\mu$ or $\gamma$), where the blue solid line corresponds to stable modes and
the red dashed line denotes unstable ones. The top and middle panels show
the power curves of $\| \Phi \|^2$, a sample profile
of the nonlinear mode $\Phi$, and the spectrum of linearization before and after the instability
bifurcation. The bottom panels show the imaginary part of eigenvalues
$\lambda$ and the Krein quantity of isolated eigenvalues. Green color corresponds to
eigenvalues $\lambda \in i \mathbb{R}$ with the positive Krein
signature, red -- to those with the negative Krein signature, and
black color is used for complex eigenvalues $\lambda \notin i \mathbb{R}$ and
for the continuous spectrum.

\begin{figure}[htbp]
\centerline{
\includegraphics[max size={1.3\textwidth}{0.95\textheight}]{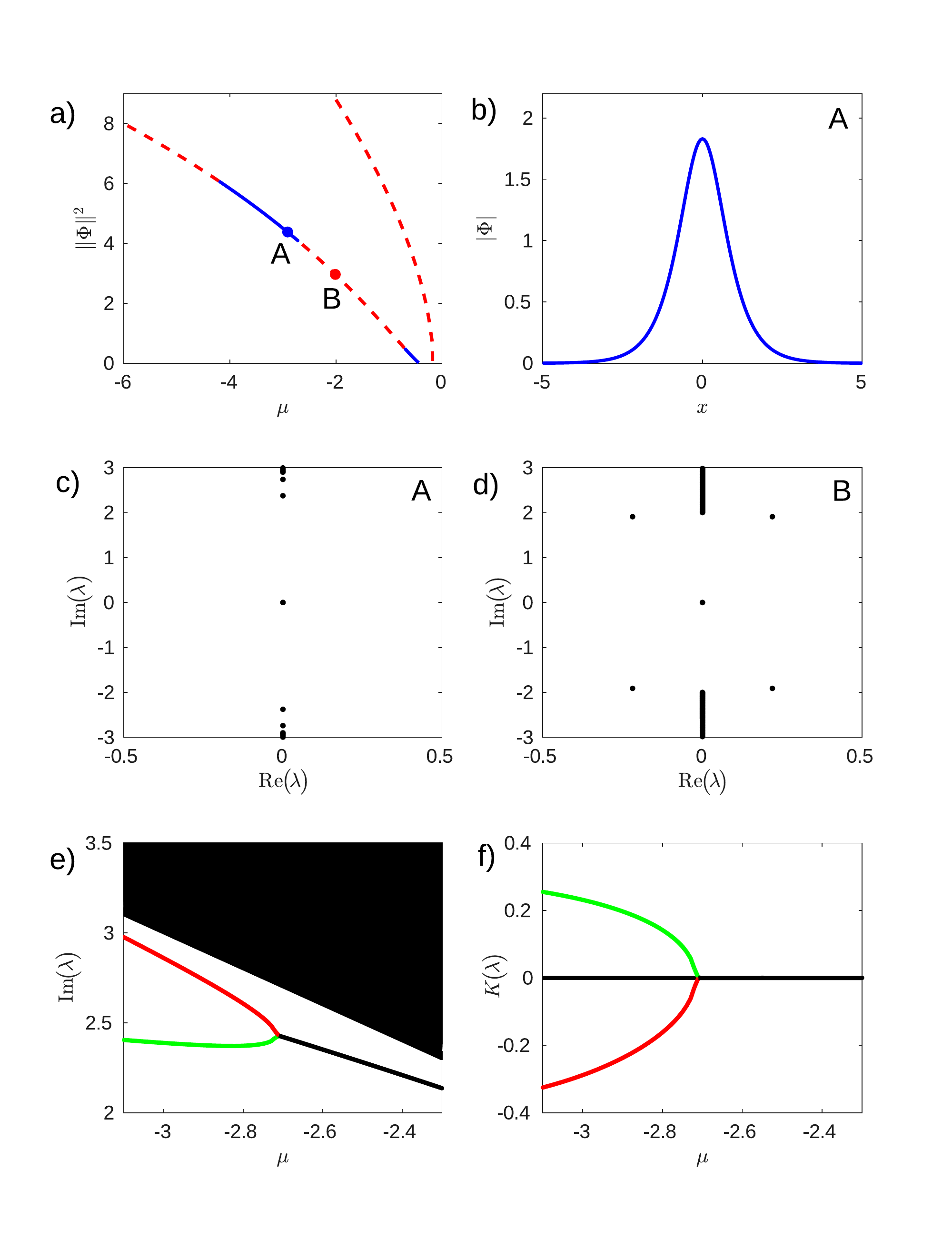}}
\caption{Scarf II potential~\eqref{pot-Wadati} with $V_0 = 2$, $\gamma = -2.21$.
    (a) Power curves versus $\mu$.
    (b) Amplitude profile for point $A$.
    (c) Spectrum of linearization for point $A$.
    (d) Same for point $B$.
    (e) ${\rm Im}(\lambda)$ for the spectrum of linearization versus $\mu$.
    (f) Krein quantities for isolated eigenvalues versus $\mu$. }
\label{fig:yang}
\end{figure}

\begin{figure}[htbp]
\centerline{
\includegraphics[max size={1.3\textwidth}{0.95\textheight}]{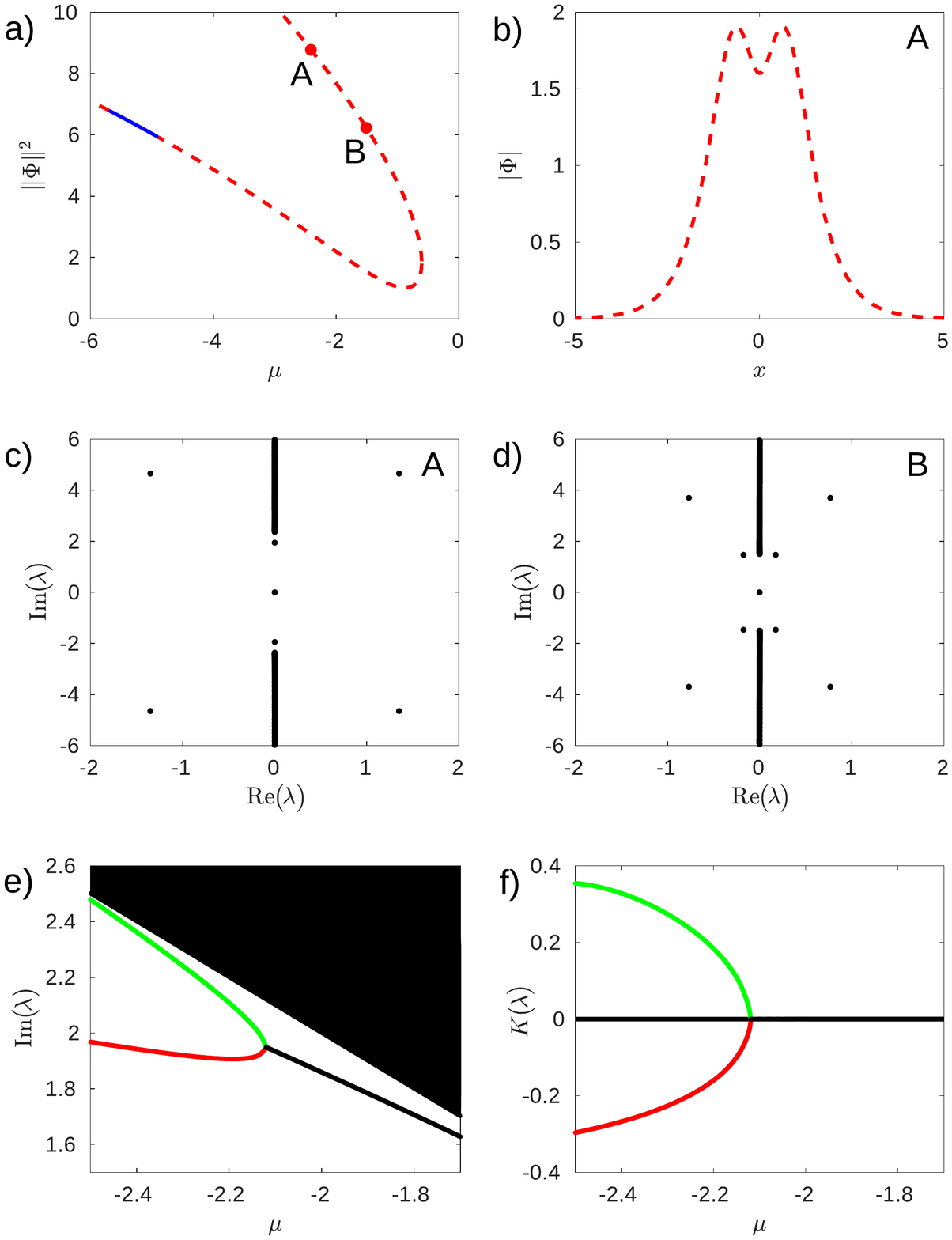}}
\caption{Scarf II potential~\eqref{pot-Wadati} with $V_0 = 3$, $\gamma = -3.7$.
    (a) Power curves versus $\mu$.
    (b) Amplitude profile for point $A$.
    (c) Spectrum of linearization for point $A$.
    (d) Same for point $B$.
    (e) ${\rm Im}(\lambda)$ for the spectrum of linearization versus $\mu$.
    (f) Krein quantities for isolated eigenvalues versus $\mu$. }
\label{fig:wadati}
\end{figure}

Figure~\ref{fig:yang} (a)-(f) shows the instability bifurcation for the Scarf II potential (\ref{pot-Wadati})
studied in~\cite{yang} in the focusing case with $g=1$. Here $V_0 = 2$, $\gamma = -2.21$, and
the first branch of the nonlinear modes $\Phi$ is considered.
As two eigenvalues with different Krein signatures coalesce, they bifurcate into
a complex quadruplet, in agreement with Theorem \ref{theorem-main}.
Note that there's a small region of stability for the nonlinear modes
$\Phi$ of small amplitudes, as it was shown in~\cite{yang}.

Figure~\ref{fig:wadati} (a)-(f) shows the instability bifurcation for the Scarf II
potential~\eqref{pot-Wadati} studied in~\cite{wadati} in the focusing case with $g=1$. 
Here $V_0=3$, $\gamma=-3.7$, and the second branch of the nonlinear modes $\Phi$
is considered. The second branch is unstable with at least one complex quadruplet
for all values of parameter $\mu$ used. The imaginary part of this complex quadruplet 
is not visible on Figure~\ref{fig:wadati} (e) as it coincides with the location 
of the continuous spectrum. In the presence of this complex quadruplet,
we observe a coalescence of two simple eigenvalues $\lambda_1,\lambda_2 \in i \mathbb{R}$
and the instability bifurcation into another complex quadruplet. Numerical evidence confirms
that the eigenvalues have the opposite Krein signatures prior to collision,
allowing us to predict the instability bifurcation, in agreement with Theorem \ref{theorem-main}.

\begin{figure}[htbp]
\centerline{
\includegraphics[max size={1.3\textwidth}{0.95\textheight}]{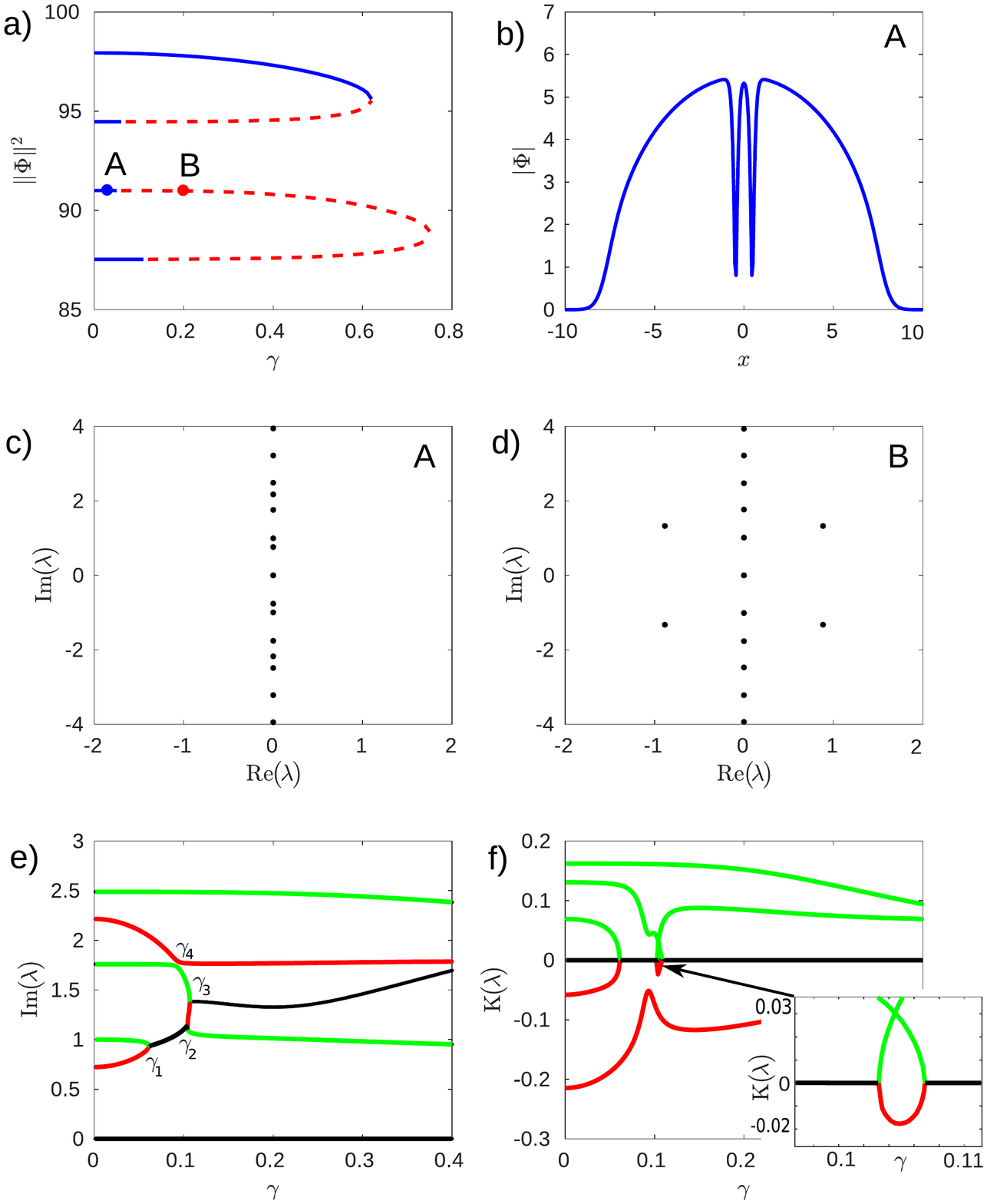}}
\caption{Confining potential~\eqref{pot-BEC}, scaled as in~\eqref{pot-BEC-scaled}.
        (a) Power curves versus $\gamma$.
        (b) Amplitude profile for point $A$.
        (c) Spectrum of linearization for point $A$.
        (d) Same for point $B$.
        (e) ${\rm Im}(\lambda)$ for the spectrum of linearization versus $\gamma$.
        (f) Krein quantities for isolated eigenvalues versus $\gamma$. }
\label{fig:kev3}
\end{figure}
\begin{figure}[htbp]
\centerline{
\includegraphics[max size={1.3\textwidth}{0.95\textheight}]{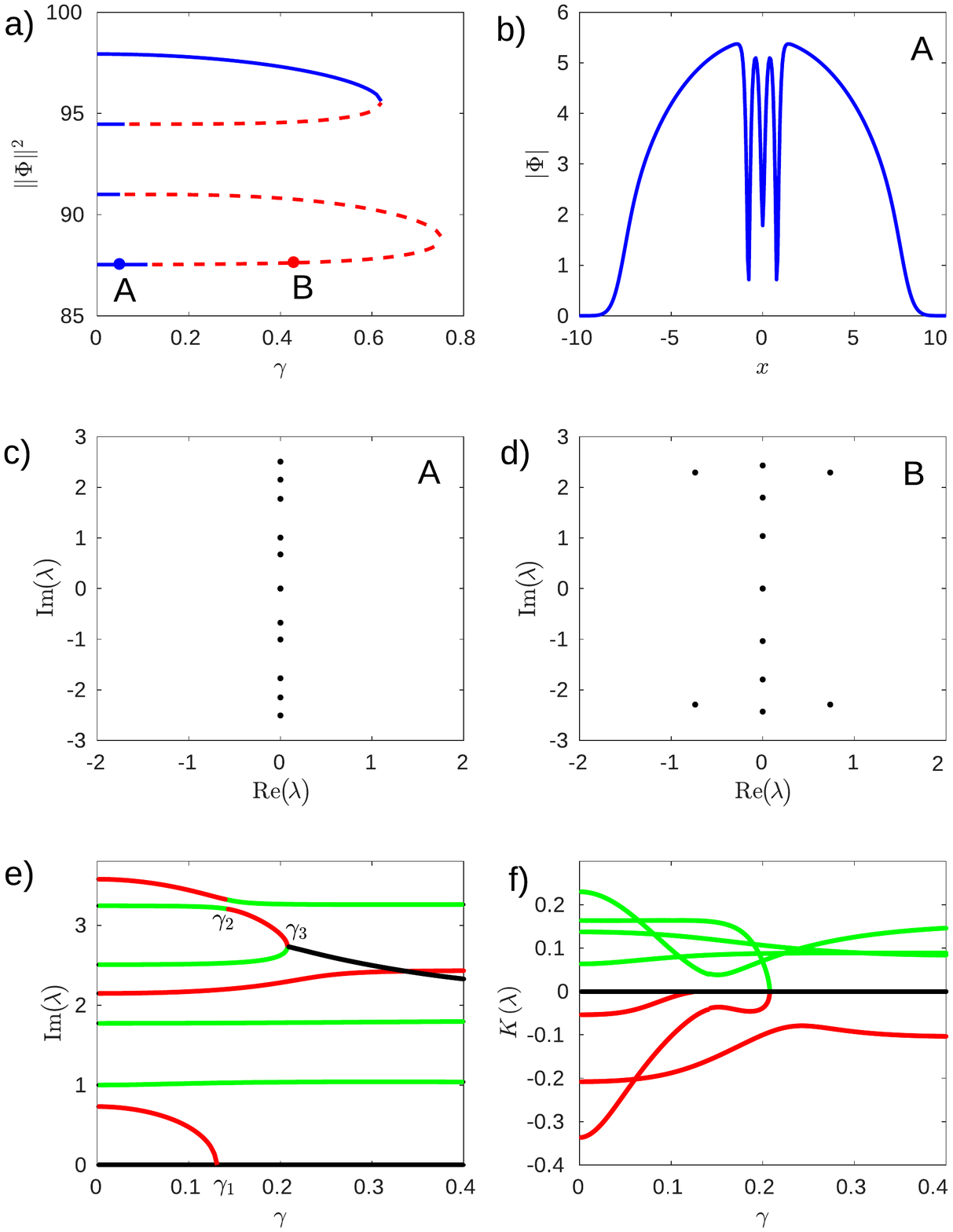}}
\caption{Confining potential~\eqref{pot-BEC}, scaled as in~\eqref{pot-BEC-scaled}.
        (a) Power curves versus $\gamma$.
        (b) Amplitude profile for point $A$.
        (c) Spectrum of linearization for point $A$.
        (d) Same for point $B$.
        (e) ${\rm Im}(\lambda)$ for the spectrum of linearization versus $\gamma$.
        (f) Krein quantities for isolated eigenvalues versus $\gamma$. }
\label{fig:kev4}
\end{figure}

Figures~\ref{fig:kev3},\ref{fig:kev4} (a)-(f) show the confining potential~\eqref{pot-BEC} studied in~\cite{kevrekidis},
in the defocusing case with $g=-2$. Compared to~\eqref{pot-BEC}, we use a scaled version of this
potential to match the one in~\cite{kevrekidis}:
\begin{equation}
    \label{pot-BEC-scaled}
V(x) = x^2, \quad W(x) = 2\Omega^{-3/2}xe^{-\frac{x^2}{2\Omega}},
\end{equation}
where $\Omega=10^{-1}$ is a scaling parameter.
There are four branches of the nonlinear modes $\Phi$ shown,
out of which we highlight only the third and fourth branches.
The first branch is stable, whereas the second branch becomes unstable because of a coalescence
of a pair of eigenvalues $\pm \lambda \in i \mathbb{R}$ with the negative Krein signature
at the origin \cite{kevrekidis}. The third and fourth branches are studied in Figures~\ref{fig:kev3} and~\ref{fig:kev4}.

In Figure~\ref{fig:kev3} we can see that there are three bifurcations occurring at $\gamma_1\approx 0.07$,
$\gamma_2 \approx 0.1031$ and $\gamma_3\approx 0.1069$. For each bifurcation two eigenvalues with different Krein
signatures collide and bifurcate off to the complex plane in accordance with Theorem \ref{theorem-main}.
In addition, two simple eigenvalues with
different Krein signatures nearly coalesce  near $\gamma_4 \approx 0.1$.
Figure~\ref{fig:comp} (a) shows the norm
of the difference between the two eigenvectors and two adjoint eigenvectors for the two simple eigenvalues
while $\gamma$ is increased towards $\gamma_4$. As the difference does not vanish,
we rule out this point as the bifurcation point for the defective eigenvalue. Consequently,
the eigenvalues are continued past this point with preservation of their Krein signatures.

In Figure~\ref{fig:kev4} we can see three bifurcations occurring at $\gamma_1\approx 0.1303$,
$\gamma_2\approx 0.1427$, and $\gamma_3\approx 0.2078$. At $\gamma_1$,
an eigenvalue pair with negative Krein signature coalesce at zero and become a pair of
real (unstable) eigenvalues. As $\gamma$ is increased towards $\gamma_2$,
two eigenvalues with opposite Krein signature move towards each other.
Figure~\ref{fig:comp} (b) illustrates that the norm of the difference between the two eigenvectors
and the two adjoint eigenvectors vanishes at the coalescence point.
Therefore, we conclude that at $\gamma_2$ we have a defective eigenvalue
which does not split into a complex quadruplet. According to Theorem \ref{theorem-main},
the defective eigenvalue does not split into complex unstable eigenvalues only if
 the non-degeneracy condition~\eqref{non-degeneracy} is not satisfied. Similar safe passing
of eigenvalues of opposite Krein signature through each other is observed in~\cite{yang}.
The behavior near $\gamma_2$ shows that having opposite Krein signatures
prior to coalescence of two simple eigenvalues into a defective eigenvalue
is a \emph{necessary but not sufficient} condition for the instability bifurcation.
At $\gamma_3$, two eigenvalues with opposite Krein signatures coalesce and
bifurcate into a complex quadruplet according to Theorem \ref{theorem-main}.

\begin{figure}[t]
\centerline{
\includegraphics[width=0.7\textwidth]{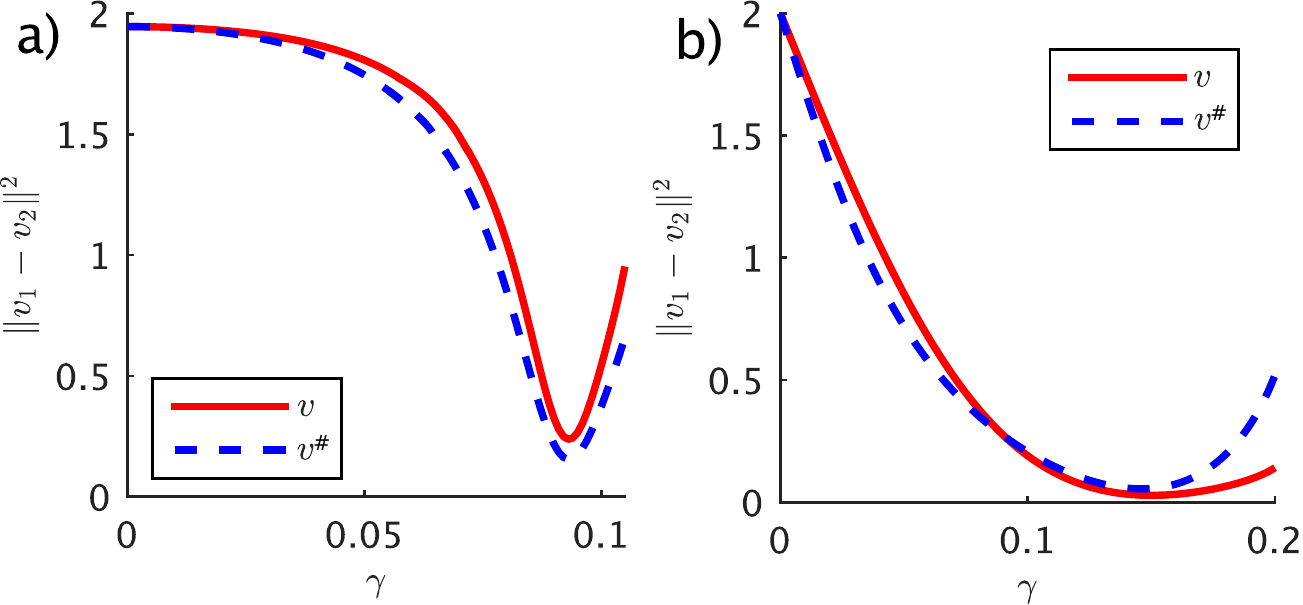}}
\caption{The norm of the difference between the two eigenvectors and the two adjoint
    eigenvectors prior to a possible coalescence point: (a) for Figure~\ref{fig:kev3}
(b) for Figure~\ref{fig:kev4}.}
\label{fig:comp}
\end{figure}

\section{Discussion}
\label{sec-conclusion}

In this work, we introduced the Krein quantity for simple isolated eigenvalues
in the linearization of the nonlinear modes in the $\mathcal{PT}$-symmetric NLS equation.
We proved that the Krein quantity is zero for complex eigenvalues and nonzero for simple
purely imaginary eigenvalues. When two simple eigenvalues coalesce
on the imaginary axis in a defective eigenvalue,
the Krein quantity vanishes and we proved
under the non-degeneracy assumption
that this bifurcation point produces
complex unstable eigenvalues on one side of the bifurcation point.
This result shows that the main feature of the instability bifurcation
in Hamiltonian systems is extended to the $\mathcal{PT}$-symmetric NLS equation.

There are nevertheless limitations of this theory in the $\mathcal{PT}$-symmetric systems.
First, the adjoint eigenvectors are no longer related to the eigenvectors of the spectral
problem, which opens up a problem of normalizing the adjoint eigenvector relative to the
eigenvector. We fixed the sign of the adjoint eigenvector in the Hamiltonian limit
and continue the sign off the Hamiltonian limit by using continuity of eigenvectors
along the parameters of the model.

Second, if the bifurcation point corresponds to a semi-simple eigenvalue, then the bifurcation theory
does not lead to the same conclusion as in the Hamiltonian case. 
The first-order perturbation theory results in the non-Hermitian matrices, hence 
it is not clear how to conclude on the splitting
of the semi-simple eigenvalues on each side of the bifurcation point.

Finally, coalescence of the simple purely imaginary eigenvalues at the origin and the related instability
bifurcations are observed frequently in the $\mathcal{PT}$-symmetric systems and they are not
predicted from the Krein quantity. Therefore, we conclude that the stability theory
of Hamiltonian systems cannot be fully extended to the $\mathcal{PT}$-symmetric NLS equation,
only the necessary condition for the instability bifurcation can be, as is shown in this work.


\end{document}